\documentclass[conference]{IEEEtran}
\usepackage{amsmath}
\usepackage{amssymb}
\usepackage{floatfig,epsfig}
\usepackage{verbatim}
\usepackage{pifont}
\usepackage{algorithm}
\usepackage{algorithmic}
\usepackage{alltt}
\usepackage{bm}
\usepackage{isolatin1}
\usepackage{graphics}
\usepackage{epic}
\usepackage{eepic}
\usepackage{graphpap}
\usepackage{psfrag}
\usepackage{float}
\usepackage{subfigure}

\newtheorem{thm}{Theorem}

\newtheorem{pro}{Proposition}

\newtheorem{lem}{Lemma}

\ifCLASSINFOpdf
\else
\fi
\begin{document}
%
% paper title
% can use linebreaks \\ within to get better formatting as desired
\title{Decentralized Minimum-Cost Repair for Distributed Storage Systems}

% author names and affiliations
% use a multiple column layout for up to three different
% affiliations
\author{Majid Gerami,  Ming Xiao, Carlo Fischione, Mikael Skoglund,
\\  ACCESS Linnaeus Centre, Royal Institute of Technology, KTH, Sweden, \\ E-mail: \{gerami, mingx, carlofi, skoglund\}@kth.se\\}
\date{\today}

\maketitle

\begin{abstract}
%\boldmath
There have been emerging lots of  applications for distributed storage systems e.g., those in wireless sensor networks or cloud storage. Since storage nodes in wireless sensor networks  have limited battery,
it is valuable to find a repair scheme with optimal transmission costs (e.g., energy). The
optimal-cost repair has been recently investigated in a
centralized way. However a centralized control
mechanism may not be available or is very expensive. For the scenarios, it is
interesting to study optimal-cost repair in a decentralized setup.
We formulate the optimal-cost repair as convex optimization
problems for the network with convex transmission costs. Then we
use primal and dual decomposition approaches to decouple the
problem into subproblems to be solved locally. Thus, each
surviving node, collaborating with other nodes, can minimize its
transmission cost such that the global cost is minimized. We
further study the optimality and convergence of the algorithms.
Finally, we discuss the code construction and determine the field size for
finding feasible network codes in our approaches.
\end{abstract}
% IEEEtran.cls defaults to using nonbold math in the Abstract.
% This preserves the distinction between vectors and scalars. However,
% if the conference you are submitting to favors bold math in the abstract,
% then you can use LaTeX's standard command \boldmath at the very start
% of the abstract to achieve this. Many IEEE journals/conferences frown on
% math in the abstract anyway.

% no keywords

% For peer review papers, you can put extra information on the cover
% page as needed:
% \ifCLASSOPTIONpeerreview
% \begin{center} \bfseries EDICS Category: 3-BBND \end{center}
% \fi
%
% For peerreview papers, this IEEEtran command inserts a page break and
% creates the second title. It will be ignored for other modes.
\IEEEpeerreviewmaketitle

\section{Introduction}
% no \IEEEPARstart
Wireless sensor networks consist of several small devices (e.g., sensors) which measure or detect a physical quantity of interest e.g., temperature, dust, light and so on. The main characteristics of these sensors are on limited battery, low CPU power, limited communication capability and small memory \cite{Kong01}. These nodes are often vulnerable. Thus to make the data reliable over these unreliable node, the data can be encoded and distributed among small storage devices \cite{Kong01}, \cite{Dimk01}, \cite{Dimk02}.
When a storage node fails, to maintain the reliability of
systems, an autonomous algorithm should regenerate a new storing
node. The process is generally known as repair. Repair process
will cause traffic and transmission cost. The repair process with
the aim of minimizing traffic leads to the proposal of optimal bandwidth (traffic)
regenerating codes \cite{Dimk01}. The repair with the objective of
minimizing transmission costs leads to the minimum-repair-cost
regenerating codes in e.g., \cite{Maj01}.

The regenerating code \cite{Dimk01} in a distributed storage
system with $n$ nodes is actually a type of erasure codes by which
any $k$ ($k \leqslant n$) out of $n$ nodes can reconstruct the
original file. This property, called the regenerating code
property (RCP), is desirable since it is optimal in
providing reliability using a given amount of storage. In the repair process, the new
node may not have the same coded symbols as the lost node. However
it preserves the RCP. This type of repair is known as functional
repair. Reference \cite{Dimk01} also models distributed storage
systems and the repair process by an acyclic directed graph,
namely,\textit{ information flow graph}. The graph involves three types of
nodes: a source node, storage nodes, and a data collector.  When a
node fails, surviving nodes send $\gamma$ bits of coded symbols to
the new node. Cut analysis on the information flow graph shows the
fundamental storage-bandwidth tradeoff. In \cite{Wu01}, it is
shown that the tradeoff can be achieved by deterministic/random
linear network codes (\cite{Ho01}). In \cite{Kong01} and \cite{Dimk01}, decentralized
approaches for erasure code construction has been proposed respectively based on fountain code and
random linear network coding.

%In order to reduce the repair traffic, reference
%\cite{Yuchong01} proposes  failed nodes cooperation. This idea
%later has been extended to functional and exact repair in
%\cite{Ken01}.

Reference \cite{Maj01} seeks to minimize repair-cost with the RCP
preserved. Furthermore, surviving node cooperation (SNC) is also
proposed in \cite{Maj01}. That is, a surviving node can combine
the data from other surviving nodes and its own data. The
transmission cost is optimized for linear costs with a central
controlling way. Here we shall study the process of
optimal-cost-repair in a decentralized method. The scenario is
interesting when the central control is difficult or expensive. For instance, a centralized control in distributed storage in wireless sensor networks is difficult or even impossible. To
achieve a decentralized method in minimum-cost repair, we first formulate problems as convex optimization
problems. Then we study decentralized methods for finding an
optimal-cost subgraph decoupled from code construction. For
the purpose, we present two distributed algorithms based on primal
and dual decomposition. With the minimum-cost subgraph, we show
that there exists a code over a finite field to regenerate the new
node properly.

The rest of the paper is organized as follows. We formulate the
minimum-cost repair problem in Section \ref{problemForm}. Then,
Section \ref{DectralizedOpt}  presents primal and dual
decomposition algorithms for finding minimum-cost repair subgraph
in a distributed way. We discuss in Section
\ref{sec:code-construction} the issue of the code construction and
required field sizes.

\section{Problem Formulation} \label{problemForm}
Consider a network with $n$ nodes. There are paths
connecting nodes. We denote the transmission cost from
node $i$ to node $j$ by function $f_{ij}$. We only consider the
convex cost. Thus, if $z_{ij}$ is the number of bits (packets)
transmitting from node $i$ to $j$, $f_{ij}$ is a convex function
of $z_{ij}$. We assume that each node knows the cost of links to
its neighbor in the network. For simplicity, we assume that the
network is delay-free and acyclic. In what follows, we first
present the modified information flow graph to analyze the repair process.

\subsection{Modified Information Flow Graph }\label{ModifiedIFG}

Consider a storage system with the source original file of size
$M$ distributed among $n$ nodes in which each node stores $\alpha$
units and any $k$ out of $n$ nodes can rebuild the original file.
We denote the source file with an $M\times1$ vector
$\underline{s}$. Then, the code on node $i$ can be evaluated by a
matrix
 $\underline{Q}_i=(\underline{q}_i^1,\cdots,\underline{q}_i^{\alpha})$ of size $M\times\alpha$ where each
 column ($\underline{q}_i^j$) represents the code coefficients of  fragment $j$
 on node $i$. The stored data in node $i$ is $\underline{X}_i=\underline{Q}_i^{T}\underline{s}$.
Then we can denote the flow of information (and topology of
networks) in a distributed storage system by a directed acyclic
graph denoted as $G(n,k,\alpha)=(N,A)$, where $N$ is the set of
nodes and $A$ is the set of directed links.

Similar to \cite{Dimk01}, graph $G(n,k,\alpha)$ consists of three
different types of nodes: a source node, storage nodes and data
collector ($DC$). The source node contains the original file which
is going to be distributed among storage nodes; The storage nodes
consists of two kinds of nodes, namely, $in$ and $out$ nodes with
a link of capacity $\alpha$ (the storage size) between them; The
data collector can reconstruct the original file by connecting to
$k$ $out$ nodes. Yet different from \cite{Dimk01}, the modified
flow graph shall reflect the topology of the network. Thus, there
might not exist direct channels (edges in $G$) from a surviving
node to the new node. A storage node may have to forward the data
of other nodes to the new node, depending on the network topology.
When a node fails, all the surviving nodes ($n-1$ nodes) can join
the repair process. An optimization algorithm shall determine the
optimal traffic on the links and hence the number of nodes for
repair. An example of the modified information flow graph for a
distributed storage system of a four-node tandem network is given
in Fig.~\ref{modifiedgraph}, where node $4$ fails and node $5$ is
the new node.
\begin{figure}%[b]
 \centering
 \psfrag{s}[][][1.5]{ S }
 \psfrag{a}[][][1.5]{ $\alpha$ }
 \psfrag{i}[][][1.5]{ $\infty$ }
 \psfrag{m1}[][][1.5]{ node 1 }
 \psfrag{m2}[][][1.5]{ node 2 }
 \psfrag{m3}[][][1.5]{ node 3 }
 \psfrag{m4}[][][1.5]{ node 4 }
 \psfrag{m5}[][][1.5]{ node 5 }
 \psfrag{z12}[][][1.5]{ $z_{(12)}$ }
 \psfrag{z23}[][][1.5]{ $z_{(23)}$ }
  \psfrag{z35}[][][1.5]{ $z_{(35)}$ }

 \resizebox{8cm}{!}{\epsfbox{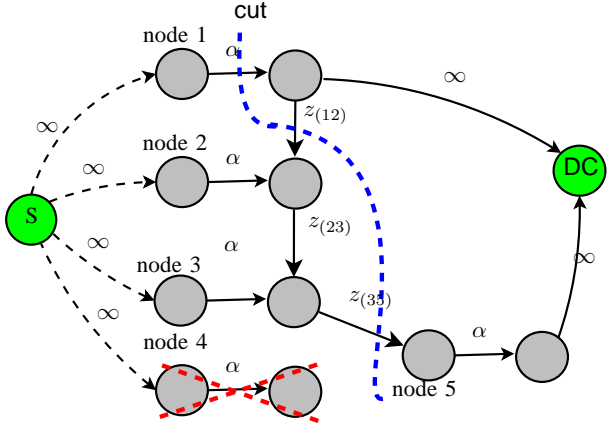}}
\caption{Modified information flow graph for a four-node tandem
network. There are directed channels connecting node 1 to node 2,
node 2 to node 3, and node 3 to node 4, respectively. Node 4 fails
and node 5 is the new node.}
 \label{modifiedgraph}
\end{figure}

For analysis, we use a column vector to denote the number of
fragments transmitted on the links of the network. The vector is
termed as \emph{subgraph}
($\underline{z}=[z_{(ij)}]_{\mid_{(ij)\in A}}$). For
a given network, our objective is to minimize the cost
($\sigma_c)$ during the repair process. With the subgraph
$\underline{z}=[z_{(ij)}]_{\mid_{(ij)\in A}}$,
 and cost function $f_{ij}$, the repair cost is
\begin{equation}
\sigma_c \triangleq \sum_{(ij)\in A} f_{ij}(z_{(ij)}).
\end{equation}

\subsection{Constraint Region} \label{constregion}

In the repair process, it is required that any $k$ nodes can
reconstruct the original file. This property is known as the
regenerating code property (RCP). In the literature, the process
that a node fails and a new node is regenerated is called a stage
of repair. The RCP must be preserved in any stage of repair. Thus,
in the repair process we should have the RCP for the system with
the new node and surviving nodes.  Hence, any cut in the modified
information graph must not be less than $M$, i.e., the original
file size. The requirement is called the cut constraint. Thus, we
should find the minimum $\sigma_c$ under the cut
 constraints. Since there are multiple cuts in the networks, there
will be multiple cut constraints. If we assume $R$ constraints,
the constraints represent the feasible region in our problem. We
call the region polytope $\Psi$, which can be denoted by the
following $R$ linear inequalities,
\begin{equation}
  \sum_{(ij)\in A} h_{(ij)}^r(z_{(ij)})\leq 0 \\ \text{ for } r=1,\cdots, R,
 \end{equation}
where $h_{(ij)}^r(z_{(ij)})$ is an affine function of $z_{(ij)}$
in the $r$-th constraint.

The polytope $\Psi$ is restricted by linear inequalities. Hence,
if $z_{(ij)}$s are real numbers then the constraint region $\Psi$
is convex. We can reasonably assume that $z_{(ij)}$s are real
numbers. Note that the file is measured by bits but it is normally
quite large. Thus we can consider $z_{(ij)}$ real valued.
Following this assumption, $\Psi$ constitutes a convex region.
Since the constraint region is convex, whenever the cost function
is convex, the problem is convex.

\subsection{Convex Optimization} \label{sec:linearOpt}

With the constraint region and objective function, we can
formulate the optimization problem as follows,
\begin{equation}
 \begin{array}{lc}
 \mbox{\text{minimize}} &  \sum_{(ij)\in A} f_{ij}(z_{(ij)}) \\
 \mbox{\text{subject to}} & \sum_{(ij)\in A} h_{(ij)}^r(z_{(ij)})\leq 0 \\ & \text{ for } r=1,\cdots, R,  \\
 & z_{(ij)}\geq $0$.
 \end{array}
 \label{opt-lin_A}
 \end{equation}

Problem (\ref{opt-lin_A}) can be solved centrally as in
\cite{Maj01} if there is a central control mechanism. Consequently
the optimal cost subgraph can be found. Without central control
schemes, we can find the optimal cost subgraph in a decentralized
manner as follows. Corresponding to the minimum cost subgraph for
$\alpha=M/k$,
 we can also find a decentralized coding
scheme (e.g., random linear network codes) for the repair
 satisfying RCP (to be shown in Section
\ref{sec:code-construction}).

\section{Minimum-Cost Subgraph by Decentralized Algorithms} \label{DectralizedOpt}

We first show problem (\ref{opt-lin_A}) can be separated to
$(n-1)$ subproblems. To decouple the problem into subproblems, we
apply primal and dual decomposition methods \cite{OptLec}. These
approaches lead us to distributed algorithms of finding the
optimal-cost repair subgraph. Further we analyze their properties
and evaluate their performance.

\subsection{Primal Decomposition}\label{PrimDecom}
The cost function of problem (\ref{opt-lin_A}) can be decoupled
into $n-1$ parts, each associated to a surviving node. Then every
node solves an optimization problem locally and a master node
coordinates the problem solving (we shall show that this master
problem can be solved in a decentralized way with communication
between nodes). Without loss of generality, we assume node $1$
fails. For decomposition, we rewrite the problem (\ref{opt-lin_A})
as the following form,
\begin{equation}
 \begin{array}{lc}
 \mbox{\text{minimize}} &  \sum_{i=2}^n \sum_{j \mid (ij)\in A} f_{ij}(z_{(ij)}) \\
 \mbox{\text{subject to}} & \sum_{i=2}^n \sum_{j \mid (ij)\in A} h_{(ij)}^r(z_{(ij)})\leq 0 \\ & \text{ for } r=1,\cdots, R,  \\
 & z_{(ij)}\geq $0$.
 \end{array}
 \label{opt-lin}
 \end{equation}

 Then, using primal decomposition with a constraint \cite{OptLec}, each nodes minimizes its transmission cost by,
\begin{equation}
 \begin{array}{lc}
 \mbox{\text{minimize}} &   \sum_{\{j \mid (ij)\in A\}}f_{ij} (z_{(ij)})  \\
 \mbox{\text{subject to}} &   \sum_{\{j \mid (ij)\in A\}} h_{(ij)}^r(z_{(ij)})\leq t_i^r\\&\text{ for } r=1,\cdots,R.\\
 & z_{(ij)}\geq $0$.
 \end{array}
 \label{sub1}
 \end{equation}

%Node $n$ optimizes %\textbf{Why node n is different from others?},
% \begin{equation}
% \begin{array}{lc}
% \mbox{\text{minimize}} &   \sum_{\{j \mid (ij)\in A\}}f_{ij} (z_{(ij)})  \\
% \mbox{\text{subject to}} &   \sum_{\{j \mid (ij)\in A\}} h_{(ij)}^r(z_{(ij)})\leq -(t_2^r+...\\ & t_3^r+\cdots+t_{(n-1)}^r)\\&\text{ for } r=1,\cdots,R\\
% & z_{(ij)}\geq $0$.
% \end{array}
% \label{sub2}
% \end{equation}

 Finally the following master problem iteratively update parameters: $ t_2^1,\cdots,t_2^R, t_3^1, \cdots, t_i^r, \cdots, t_n^R $
 \begin{eqnarray}
  \text{minimize}    & \phi  =  \phi_2(t_2^1,t_2^2,\cdots,t_2^R )+ \nonumber \\ & \cdots  +\phi_n(t_n^1,t_n^2,\cdots,t_n^R),\nonumber \\
  \text{subject to} & t_2^r+...+t_3^r+\cdots+t_{n}^r=0,
   \label{masterProblem}
 \end{eqnarray}
 where for each node, $\phi_i(t_i^1,t_i^2,\cdots,t_i^R )$ is calculated using the Lagrange dual function, associating $\lambda_i^1,\cdots,\lambda_i^R$ as Lagrangian variables of $R$ inequality constraints in subproblem $i$, as
 \begin{eqnarray}
 \phi_i(t_i^1,t_i^2,\cdots,t_i^R )=  \sup_{\lambda_i^1,\cdots,\lambda_i^R} \inf_{z_{(ij)} \mid (ij) \in A}  f_i(z_{(ij)})
 \nonumber \\ -\lambda_i^1 (h_i^1(z_{(ij)})-t_i^1)-\cdots-\lambda_i^R (h_i^R(z_{(ij)})-t_i^R).
  \label{phi}
 \end{eqnarray}
 We can relax the constraint in (\ref{masterProblem}) by setting $t_{n}^r=-(t_2^r+...\\  t_3^r+\cdots+t_{(n-1)}^r)$ in subproblem $n$. Thus, the gradient of function $\phi( t_2^1,\cdots,t_2^R, t_2^1, \cdots, t_i^r, \cdots, t_{n-1}^R)$ in (\ref{phi}) is
  \begin{equation}
 \Delta_p=(\lambda_2^1-\lambda_n^1,\cdots,\lambda_2^R-\\
 \lambda_n^R,\cdots,\lambda_{(n-1)}^R- \\
 \lambda_n^R).
  \end{equation}
  Therefore, the iterative algorithm is

\texttt{Algorithm 1:Primal iterative algorithm}

 \textbf{Repeat:}
 \begin{enumerate}
    \item Every node  solves a subproblem\\
        Node $i$, for $2\leq i \leq n$, solves the subproblem (\ref{sub1}), finding $z_{(ij)\mid(ij) \in A }$ and $(\lambda_i^1,\cdots,\lambda_i^R) $.\\
        %Node $n$  solves the subproblem (\ref{sub2}), finding
%        $z_{(nj) \mid(nj) \in A }$ and
%        $(\lambda_n^1,\cdots,\lambda_n^R)$.
    \item Update vector $\underline{t}=(t_2^1,\cdots,t_2^R, t_2^1, \cdots, t_i^r, \cdots, t_{n-1}^R)$\\
        $\underline{t}:=\underline{t}-\alpha_k\Delta_p,$ where
        $\alpha_k$ is the iteration step length.

 \end{enumerate}

 \textbf{Until:} The stopping criterion (as follows) is satisfied.
%\textbf{How the algorithm is terminated? /Ming}

The algorithm can be stopped after passing $T$ (pre-defined)
iterations for delay sensitive conditions or after  achieving certain
level of accuracy (e.g.,
$\|\sigma_c(k)-\sigma_c(k-1)\|<\varepsilon$, where $\varepsilon$
is small and positive). The properties of Algorithm 1 are
discussed as follows.

\subsubsection{Optimality}
We know problem (\ref{opt-lin}) has feasible solutions
 (by e.g., simply assigning $z_{(ij)}=M$ all the cut constraints are satisfied). According to \cite{OptLec}, problem (\ref{opt-lin})
and the decomposed problem are equivalent. Hence, as long as the
convergence of the decomposed problem is proved, it converges to
the optimal solution.

\subsubsection{Convergence}

\begin{pro} For the decomposed problems (\ref{sub1}), (\ref{masterProblem}),  Algorithm 1 converges to the optimal solutions.
\end{pro}
Proof: The proof is similar to that in \cite{OptLec}.

\subsubsection{Implementing Algorithm 1 in a decentralized way}
It seems that Algorithm 1 is still not fully decentralized since a
node is needed to solve the master problem. However, by checking
the master updating equation, we see that the equation can be
broken into $n-1$ parts if nodes can communicate to each other.
That is,
\begin{equation}
\Delta_p=(\Delta_{p2},\cdots,\Delta_{pi},\cdots,\Delta_{pn}),
\end{equation}
where for node $i$, $0 \leq i \leq (n-1)$,
\begin{equation}
 \Delta_{pi}=(\lambda_i^1-\lambda_n^1,\lambda_i^2-\lambda_n^2,\cdots,\lambda_i^R-\lambda_n^R).
\end{equation}
Consequently, at the end of each iteration,  node $i$, receives
$(\lambda_n^1,\cdots,\lambda_n^R)$ and updates its master equation
as,
 \begin{equation}
 t_i^r=t_i^r-\alpha_k \Delta_{pi}(r), \text{ for } r=1,\cdots,R.
 \end{equation}
 Node $i$ also sends the updated results to node $n$. Since we assume there exists a path between any
 pair of nodes, nodes can thus communicate and update their master equations.

 \subsection{Dual Decomposition}\label{DualDecom}
 For dual decomposition, we can compute the dual function of the optimization problem (\ref{opt-lin}), and then
 decouple the problem into $(n-1)$ subproblems as follows
\begin{eqnarray*}
&&g(\underline{\lambda},\underline{z})  = \\
&&\sum_{i=2}^{n} \sum_{\{j \mid (ij)\in A\}} f_{ij}(z_{(ij)}) -  \lambda^1( \sum_{i=2}^{n} \sum_{\{j \mid (ij)\in A\}} h_{(ij)}^1(z_{(ij)}))\\
&-& \cdots  - \lambda^R(\sum_{i=2}^{n}  \sum_{\{j \mid (ij)\in A\}} h_{(ij)}^R(z_{(ij)})) \\
&=&\sum_{i=2}^{n} (\sum_{\{j \mid (ij)\in A\}}c_{(ij)} z_{(ij)} - \sum_{r=1}^R \lambda^r \sum_{\{j \mid (ij)\in A\}} h_{(ij)}^r(z_{(ij)})),
   \label{dualf}
 \end{eqnarray*}
%\begin{equation}
%\begin{aligned}
%g(\underline{\lambda},\underline{z})= \sum_{i=2}^{n} \sum_{\{j \mid (ij)\in A\}} c_{(ij)} z_{(ij)} - \lambda^1( \sum_{i=2}^{n} \sum_{\{j \mid (ij)\in A\}} h_{(ij)}^1(z_{(ij)}))- \cdots-\lambda^M(\sum_{i=2}^{n} \sum_{\{j \mid (ij)\in A\}} h_{(ij)}^M(z_{(ij)}))= \sum_{i=2}^{n} (\sum_{\{j \mid (ij)\in A\}}c_{(ij)} z_{(ij)}- \sum_{m=1}^M \sum_{\{j \mid (ij)\in A\}} h_{(ij)}^m(z_{(ij)}))
%\end{aligned}
%\end{equation}
where $\lambda^1,\cdots,\lambda^R)$ are associated Lagrangian variables of $R$ inequalities in problem (\ref{opt-lin_A}).
Therefore, the optimization problem  can be solved distributed by
$(n-1)$ surviving nodes, where node $i$, $2 \leq i \leq n$, solves
the following problem
\begin{equation}
 \begin{array}{lc}
 g(\underline{\lambda})=\min_{z_{ij} \mid (ij)\in A} &   \sum_{\{j \mid (ij)\in A\}}f_{(ij)} (z_{(ij)})-\sum_{r=1}^R \\ & \lambda^r\sum_{\{j \mid (ij)\in A\}} h_{(ij)}^r(z_{(ij)}))
  \end{array}.
 \label{dual_sub}
 \end{equation}
Vector $\underline{\lambda}=(\lambda^1,\cdots,\lambda^R)$ is updated after each iteration in
order to minimize the duality gap by
 \begin{equation}
 \begin{array}{lc}
 \max_{\underline{\lambda}}   g(\underline{\lambda}).
 \end{array}
 \label{dual-master}
 \end{equation}
Since the gradient of $q(\underline{\lambda})$ with respect to the
variable $\lambda^r$ is:
\begin{equation}
 g'_r=\frac{\partial g}{\partial \lambda^r}=-\sum_{i=2}^{n} \sum_{\{j \mid (ij)\in A\}} h_{(ij)}^r(z_{(ij)}),
 \label{gradient_dual}
 \end{equation}
the iterative algorithm is

\texttt{Algorithm 2:Dual iterative algorithm}

 \textbf{Repeat:}
  \begin{enumerate}
    \item Every node  solves a minimization problem
        (\ref{dual_sub}), resulting in $z_{(ij)\mid(ij) \in A
        }$.
    \item Update vector $\underline{\lambda}=(\lambda^1,\cdots,\lambda^R)$\\
        $\lambda^r:=\lambda^r+\alpha_k g'_r,$ where $\alpha_k$
        is iteration step length.
 \end{enumerate}

\textbf{Until:} The stopping criterion (as Algorithm 1) is
satisfied.

 We discuss the properties of Algorithm 2 as follows,

 \subsubsection{Optimality}
 For a convex cost function,
 since the constraints in problem  (\ref{opt-lin}) are  non-strict linear inequalities, then the refined Slater condition
 is satisfied \cite{Optbook}. Therefore, strong duality holds for any convex cost in the problem (\ref{opt-lin}).

\subsubsection{Convergence}
Since Algorithm 2 uses a gradient method, it is straightforward to
show the convergence \cite{Optbook}.
\subsubsection{Implementing Algorithm 2 in a decentralized way}
Similar to Algorithm 1, the update equation can be decoupled to $(n-1)$ parts.
\subsection{Numerical results}
For illustration, we apply the decentralized algorithms for a
$4$-node tandem network in Fig. \ref{modifiedgraph} and a
$2\times3$ grid networks in Fig. \ref{grid-NC}. Then, we
numerically compare their convergence behavior. First, we use the
distributed algorithms on a repair process of the distributed
storage system in Fig. \ref{modifiedgraph}. Consider a source file
of size $M=4$ packets is distributed among $4$ nodes such that any
$k=2$ nodes can recover the original file. Assume transmission
between neighboring nodes leads to one unit cost
($f_{ij}(z_{(ij)})=z_{(ij)}$). If node $4$ fails, the optimization
problem is formulated as follows,
\begin{equation} \begin{array}{lc}
\mbox{\text{minimize}} & f(\underline{z})=z_{(12)}+z_{(23)}+ z_{(35)}  \\
\mbox{\text{subject to}} &
\begin{cases}
 z_{(35)} &\geq 2\\
 z_{(23)} &\geq 2\\
 z_{(12)}+ z_{(35)} & \geq 2.\\
\end{cases}
\end{array}.\end{equation}
If the problem can be solved centrally, the optimal approach can
regenerate the new node with $4$ units of transmission costs as in
\cite{Maj01}. Fig. \ref{distexample} compares the result of
distributed algorithms by primal and dual decomposition when
$\alpha_k=0.5/\sqrt{k}$. We can see that the primal approach has very low
convergence speed. The dual algorithm converges very fast to the
optimal value (of the centralized approach). However, the
convergence property may vary for different networks.  Consider
the example in Fig. \ref{grid-NC}. We assume $M=8$ packets are
distributed among $6$ nodes in the grid network such that any $4$
nodes can reconstruct the original file. As shown in Fig.
\ref{distexample2}, the dual algorithm converges slowly to the
optimal value of the centralized approach. Primal decomposition
has faster convergence in this network comparing to the dual algorithm. This difference might stem from the difference in their  network structure.

%Comparing robustness:

\begin{figure}%[b]
 \centering
 \psfrag{a}[][][1.5]{ $\alpha$ }
 \psfrag{m1}[][][2]{ node 1 }
 \psfrag{m2}[][][2]{ node 2 }
 \psfrag{m3}[][][2]{ node 3 }
 \psfrag{m4}[][][2]{ node 4 }
 \psfrag{m5}[][][2]{ node 5 }
 \psfrag{m6}[][][2]{ node 6 }
 \psfrag{x1}[][][2]{ $x_1$ }
 \psfrag{x2}[][][2]{ $x_2$ }
 \psfrag{x3}[][][2]{ $x_3$ }
 \psfrag{x4}[][][2]{ $x_4$ }
 \psfrag{x5}[][][2]{ $x_5$ }
 \psfrag{x6}[][][2]{ $x_6$ }
 \psfrag{x7}[][][2]{ $x_7$ }
 \psfrag{x8}[][][2]{ $x_8$ }
 \psfrag{y1}[][][2]{ $y_1$ }
 \psfrag{y2}[][][2]{ $y_2$ }
 \psfrag{y3}[][][2]{ $y_3$ }
 \psfrag{y4}[][][2]{ $y_4$ }
 \psfrag{p1}[][][2]{ $p_1$ }
 \psfrag{p2}[][][2]{ $p_2$ }
 \psfrag{p3}[][][2]{ $p_3$ }
 \psfrag{p4}[][][2]{ $p_4$ }
 \psfrag{p5}[][][2]{ $p_5$ }
 \psfrag{p6}[][][2]{ $p_6$ }
 \psfrag{p7}[][][2]{ $p_7$ }
 \resizebox{6cm}{!}{\epsfbox{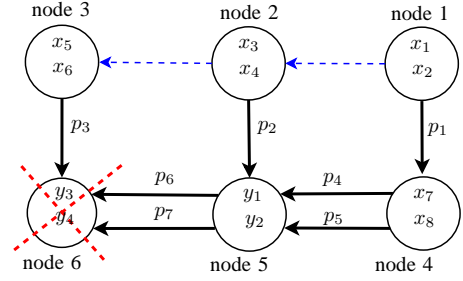}}
\caption{Optimization repair in the $2\times 3$  grid network. Each solid line represents transmission of one packet.
Dashed lines show available links which are not used in the repair
process.}
 \label{grid-NC}
\end{figure}

\begin{figure}%[b]
 \centering
 \psfrag{ylabel}[][][2]{ repair-cost($\sigma_c$) }
 \psfrag{xlabel}[][][2]{ iteration ($k$) }
 \psfrag{title}[][][2]{ Decentralized algorithms in 4-nodes tandem distributed system }
 \psfrag{data1data1data1data1}[][][1.5]{Primal decomposition }
 \psfrag{data2data2data2data2}[][][1.5]{Dual decomposition}
 \psfrag{data3data3data3data3}[][][1.5]{Optimal-cost repair }
 \psfrag{title}[][][1.5]{ Convergence of primal and dual algorithms in 4-node tandem network}
 \resizebox{9cm}{!}{\epsfbox{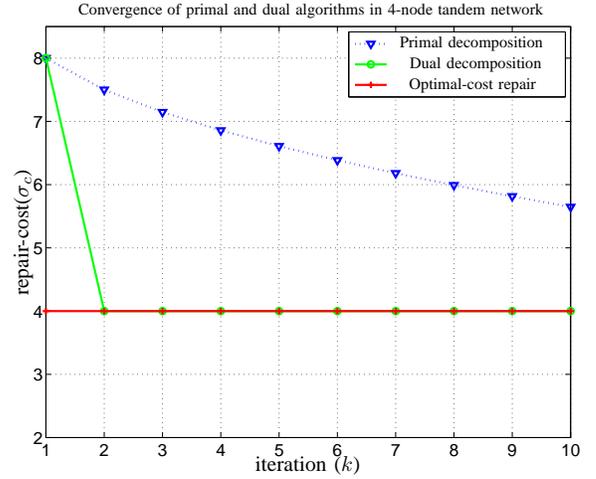}}
\caption{Distributed algorithms for finding optimal cost repair in 4-nodes tandem network. $\alpha_k=0.5/\sqrt{k}$ }
 \label{distexample}
\end{figure}
%\end{equation}

\begin{figure}%[b]
 \centering
 \psfrag{ylabel}[][][2]{ repair-cost($\sigma_c$) }
 \psfrag{xlabel}[][][2]{ iteration ($k$) }
  \psfrag{title}[][][1.5]{ Convergence of primal and dual algorithms in $2\times 3-$grid network}
% \psfrag{title}[][][2]{ Decentralized algorithms in $2\times3$-grid distributed system }
\psfrag{data1data1data1data1}[][][1.5]{Primal decomposition }
 \psfrag{data2data2data2data2}[][][1.5]{Dual decomposition}
 \psfrag{data3data3data3data3}[][][1.5]{Optimal-cost repair }
 \resizebox{9cm}{!}{\epsfbox{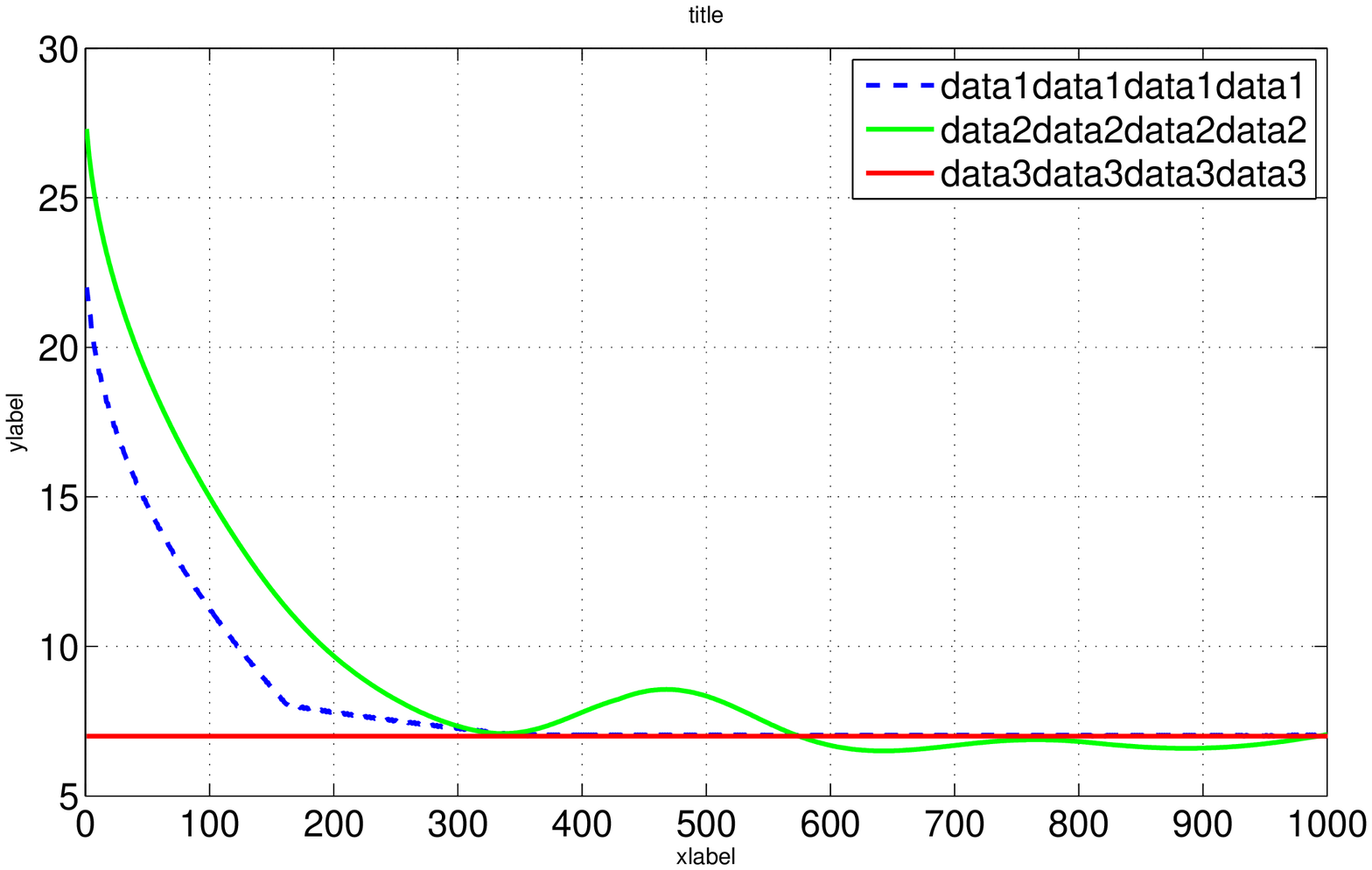}}
\caption{Distributed algorithms for finding optimal-coat repair in $2\times3$ grid network. $\alpha_k=0.5/\sqrt{k}$.}
 \label{distexample2}
\end{figure}

\section{Decentralized Optimal-Cost Minimum Storage Regenerating (OC-MSR) Code  Construction}\label{sec:code-construction}

In this section, we illustrate how to construct the regenerating code corresponding the optimal cost subgraph in Section \ref{DectralizedOpt}. In \cite{Kong01},\cite{Dimk02}, decentralized code for distributing data among storage nodes have been suggested based on rateless fountain code and linear network coding. In the repair problem, an optimum bandwidth code has been suggested by Wu \cite{Wu01}. Subsequently,  the author in \cite{Wu01} finds the the sufficient finite field size for the linear code. We try to find the optimum-cost minimum storage regenerating (MSR) code. Here
MSR means that $\alpha = M/k$. Consequently we find the required finite field size for the linear code  which regenerates the new node having RCP property (with high probability).

%in
%a randomized procedure which regenerates the new node having MDS property (with high probability).
%In \cite{Kong01},\cite{Dimk02}, decentralized code for distributing a file Here
%MSR means that $\alpha = M/k$.

To formulate the problem, suppose there is a source file of
size $M$ which is divided into $k(n-k)$ fragments and coded with a
regenerating code (satisfying the RCP) to $n(n-k)$ fragments. The
code blocks are distributed among $n$ nodes
$(\underline{Q}_1,\underline{Q}_2,\cdots,\underline{Q}_n)$. Every
node stores $\alpha= M/k=(n-k)$ fragments with the code
$(\underline{Q}_i=[\underline{q}_i^1, \underline{q}_i^2, \cdots,
\underline{q}_i^{(n-k)}])$ where $\underline{q}_i^j \in
\mathbb{F}_q^M$. When a node fails (say, $\underline{Q}_1$ fails)
the optimization algorithm finds the minimum-cost subgraph. Using
random network coding from a proper finite field guarantees the
regeneration of the new node ($\underline{Q}_1^{'}$) satisfying
the RCP. As proof, we have Lemma \ref{Lemma:code_exist} as
follows.

\begin{lem} \label{Lemma:code_exist} In the repair process of node $1$ described by
optimization problem (\ref{opt-lin}), for any selection of $k-1$
surviving nodes
($\underline{Q}_{s_1},\cdots,\underline{Q}_{s_{k-1}}$), there
exist code coefficients in which  matrix
$[\underline{Q}_1^{'},\underline{Q}_{s_1},\cdots,\underline{Q}_{s_{k-1}}]$
has full rank. That is,
\begin{equation}
\prod_{{s_1,\cdots,s_{k-1}}\subseteq {2,\cdots,n}}
\det([\underline{Q}_1^{'},\underline{Q}_{s_1},\cdots,\underline{Q}_{s_{k-1}}])\neq 0.
\end{equation} \end{lem}

\begin{proof} Own to space limitation, we skip the proof here.\end{proof}

 In  the optimal-cost repair, surviving nodes are allowed to cooperate (SNC) in order to reduce the cost as in \cite{Maj01}. Using SNC,   network coding
is also used in intermediate storing nodes. The coding process may
increase the degree of new node's polynomial considering the
determinant of coding variables \cite{KoMed}. The maximum degree
of the new node polynomial is determined by the maximum number of
times that a network coding process is used for a specific
fragment. We denote this number as $n_{nc}$. For instance, $n_{nc}=2$ in a scenario that there exist direct links from surviving nodes to new node \cite{Dimk01}, \cite{Wu01}; one step of coding in surviving node and another in new node. And, in general $n_{nc}\geq2$ in multi-hop structure using SNC, since intermediate nodes as well perform network coding on their received fragments. Thus, for a more general
scenario, we have the following result.

\begin{thm} \label{Themorem:Fieldsize} For a distributed storage system with parameters
$G(n,k,\alpha)$, and a source file of size $M$, if the finite
field is greater than $d_0$, there exists a linear network coding
such that at any stage, the RCP is satisfied, regardless of how
many failures/repairs happened before, where
$d_0=\binom{n}{k}Mn_{nc}$.
\end{thm}
\begin{proof} The proof is similar to the proof in \cite{Wu01}.\end{proof}
 %We can prove by induction.
%We start by initializing the code for $n$ nodes on which any $k$
%nodes can reconstruct the original file. Then if a node fails, the
%new node is regenerated in  a minimum cost way such that the code
%property is preserved. For induction assumption, any span of $k$
%nodes must have full rank $M$. That is,
%\begin{equation} \label{Inequl:field_size}
%\prod_{\{s_1,\cdots,s_k\}\subseteq \{1,\cdots,n\}} \det ( [
%\underline{Q}_{s_1}, \cdots, \underline{Q}_{s_k} ] ) \neq 0.
%\end{equation}
%To (\ref{Inequl:field_size}), by the spare-zero lemma, the finite
%field size ($q$) must be greater than $\binom{n}{k}S$. Since
%$n_{cc}\geq 2$ then $d_0$ is greater than $\binom{n}{k}M$.
%
%When a node fails (assume $Q_1$), the optimization algorithm finds
%the optimal-cost subgraph. Accordingly, the fragments are combined using linear network coding, and then
% the new node is regenerated. The new node ($Q_1^{'}$) beside surviving
%nodes must satisfy the RCP. Thus,
%\begin{equation}
%\prod_{{s_1,\cdots,s_{k-1}}\subseteq {2,\cdots,n}}
%\det([\underline{Q}_1^{'},\underline{Q}_{s_1},\cdots,\underline{Q}_{s_{k-1}}])\neq 0.
%\end{equation}
%By Lemma \ref{Lemma:code_exist}, the polynomial is non zero. Thus
%we can use spare-zero lemma again. Then, the finite field size
%($q$) must be greater than $\binom{n-1}{k-1}M n_{cc}$, which $d_0$
%satisfy this condition as well.

With the sufficient field size, the network codes can be easily
constructed by e.g., the random linear network coding approach
\cite{Ho01}. In summary, OC-MSR codes can be given in two steps.
First, the optimal-cost subgraph is found. It is decoupled from
coding. Then, to construct the code of the new node, network
coding coefficients are chosen (e.g., randomly) from a
sufficiently large finite field (specified by Theorem
\ref{Themorem:Fieldsize}) so that the probability of regenerating
the new node satisfying RCP would be close to 1.

\section{Conclusion}\label{Sec:conclusion}
We study a decentralized approach for optimal-cost repair in a
distributed storage system. We formulate the decentralized
optimum-cost problems as a convex optimization problems for the
network with convex transmission costs. Primal and dual
decomposition approaches are used to decouple the problem into
subproblems to be solved locally. We further study the convergence
properties of the algorithms. Numerical results show that for
tandem network, dual decomposition has much faster convergence and
for grid networks, primal decomposition is  faster. Finally, we discuss the construction of the optimal cost regenerating
codes  and discuss the
field size of the codes.


\begin{thebibliography}{1}
\bibitem{Kong01} Z. Kong, S. A. Aly, and E. Soljanin, ``Decentralized coding algorithms for distributed storage in wireless sensor networks," \textit{ IEEE Journal of Selected Areas in Communications,} vol. 28, pp. 261-268, Feb. 2010.
\bibitem{Dimk01} A. G. Dimakis, P. B. Godfrey, Y. Wu, M. J. Wainwright, and K.
 Ramchandran, ``Network coding for distributed storage systems,"
 \textit{IEEE Trans. on Info. Theory}, Sep. 2010.
\bibitem{Dimk02} A.G. Dimakis, V. Prabhakaran, K. Ramchandran, ``Ubiquitous access to
distributed data in large-scale sensor networks through decentralized
erasure codes," \textit{Proc. of IPSN,} pp. 111-117, Apr. 2005.
\bibitem{Maj01} M. Gerami, M. Xiao , and M. Skoglund, ``Optimum-cost repair in multi-hop distributed storage systems," \textit{ IEEE International Symposium on Information Theory} (ISIT) 2011.
\bibitem{Wu01}Y. Wu, ``Existence and construction of capacity-achieving network codes for distributed storage," \textit{IEEE Journal on Selected
Areas in Commun.},  vol. 28, no. 2, pp. 277-288, Feb. 2010.
\bibitem{InfoFlow} R. Ahlswede, N. Cai, S. Y. Robert Li and R. W. Yeung,
``Network information flow," \textit{IEEE Trans. on Info. Theory}, Vol. 46,
No.4, July 2000, pages 1204-1216.
\bibitem{Ho01} T. Ho, M. Médard, R. Koetter, D. R. Karger, M. Effros, J. Shi, and B.
Leong, ``A Random Linear Network Coding Approach to Multicast,"
\textit{IEEE Trans. on Info. Theory}, vol.52, pp 4413-4430, Oct.
2006.
\bibitem{Optbook} S. Boyd and L. Vandenberghe,
\textit{Convex Optimization}, Cambridge University Press, 2004.
\bibitem{OptLec} S. Boyd, L. Xiao, A. Mutapcic, \textit{Notes on decomposition methods}, Notes for EE392o, Stanford University, Oct. 2003.
\bibitem{KoMed} R. Koetter and M. Medard, ``An algebraic approach to network coding,"
\textit{IEEE/ACM Trans. Networking}, vol. 11, no. 5, pp. 782-795, Oct.
2003.
\bibitem{Yuchong01} Y. Hu, Y. Xu, X. Wang, Ch. Zhan and P. Li,
``Cooperative recovery of distributed storage systems from
multiple losses with network coding," \textit{IEEE Journal on
Selected Areas in Commun.}, Feb. 2010.
\bibitem{Ken01} K. W. Shum, ``Cooperative regenerating codes for distributed storage
systems," in IEEE Int. Conf. Comm. (ICC), Kyoto, Jun. 2011.
%\bibitem{Lun01} D. S. Lun, N. Ratnakar, M. Medard, D. Karger, T. Ho, E. Ahmad and F. Zhao,
%  ``Minimum-cost multicast over coded packet networks," \textit{IEEE Trans. on Info.
%  Theory,} vol.52, pp 2608-2623, June 2006.
\end{thebibliography}
\end{document}